\documentclass[11pt]{article}
\usepackage{fullpage}

\usepackage{epigraph}
\usepackage{nicefrac}
\usepackage[margin=1in]{geometry}
\usepackage{relsize}
\usepackage{dsfont}
\usepackage{subcaption}
\usepackage{longtable}
\usepackage{siunitx}
\usepackage{bbold}

\usepackage{latexsym}

\usepackage[dvipsnames,table]{xcolor}
\usepackage{parskip}
\usepackage{color}
\usepackage{float}
\usepackage{bbm}
\floatstyle{boxed}
\newfloat{algorithm}{t}{lop}
\usepackage{tikz}
\usepackage{varwidth}

\usepackage[titletoc,title]{appendix}
\usepackage[classfont=sanserif,langfont=roman,funcfont=italic]{complexity}
\usepackage{comment}
\usepackage{enumerate}
\usepackage{amsmath, amsthm, amssymb, amstext,  graphicx,amsopn} 

\usepackage{booktabs} 

\usepackage{nicefrac}
\usepackage{hyperref}
\hypersetup{colorlinks=true,linkcolor=black,citecolor=blue}

\definecolor{burgundy}{rgb}{0.5, 0.0, 0.13}
\definecolor{crimson}{rgb}{0.86, 0.08, 0.24}\usepackage{mathtools,extarrows}
\usepackage{url}
\newtheorem{definition}{Definition}

\newtheorem{proposition}{Proposition}

\allowdisplaybreaks

\allowdisplaybreaks

\usepackage{thmtools} 
\usepackage{thm-restate}
\usepackage[capitalise,nameinlink]{cleveref}

\crefname{figure}{Figure}{Figure}

\usepackage{mathrsfs}
\hypersetup{colorlinks=true,linkcolor=blue,citecolor=blue}
\definecolor{mygreen}{rgb}{0.0, 0.55, 0.0}
\definecolor{blue-violet}{rgb}{0.54, 0.17, 0.89}

\usepackage{tikz}
\usetikzlibrary{calc, graphs, graphs.standard, shapes, arrows, arrows.meta, positioning, decorations.pathreplacing, decorations.markings, decorations.pathmorphing, fit, matrix, patterns, shapes.misc, tikzmark}

\usepackage{amsfonts}

\definecolor{airforceblue}{rgb}{0.36, 0.54, 0.66}
\definecolor{darkblue}{rgb}{0.0, 0.0, 0.55}

\usepackage{epstopdf}
\usepackage{enumitem}
\usepackage{algorithm}
\usepackage{algpseudocode}
\usepackage{comment}
\usepackage{xcolor}
\usepackage[T1]{fontenc} 

\usepackage{multirow}
\usepackage{makecell}

\title{Computing Envy-Free up to Any Good (EFX) Allocations via Local Search}
\author{Simina Br\^anzei\footnote{Purdue University. Work done in part at Google. This research was supported in part by US National Science Foundation  grant CCF-2238372. E-mail: simina.branzei@gmail.com. The implementation of the algorithm and the scripts with data used to generate the images are available at \url{https://github.com/domnica/EFX}.}}

\begin{document}
\maketitle

\begin{abstract}
We present a simple local search algorithm for computing EFX (envy-free up to any good) allocations of $m$ indivisible goods among $n$
agents with additive valuations. EFX is a compelling fairness notion, and whether such allocations always exist remains a major open question in fair division.

Our algorithm employs simulated annealing with the total number of EFX violations as an objective function together with a  single-transfer neighborhood structure to move through the space of allocations.
It found an EFX allocation in all the instances tested, which included thousands of randomly generated inputs, and scaled  to settings with hundreds of agents and/or thousands of items. 
The algorithm's simplicity, along with its strong empirical performance makes it a simple   benchmark for evaluating future approaches.

On the theoretical side, we provide a potential function for identical additive valuations, which ensures that any strict-descent procedure under the single-transfer neighborhood ends at an EFX allocation. This represents an alternative proof of existence for identical valuations.
\end{abstract}


\section{Introduction.}

We study the fair division of a set of indivisible goods among agents with additive valuations. Formally, let $M = \{1, \ldots, m\}$ denote the set of goods and $N = \{1, \ldots, n\}$ the set of agents. Each agent $i \in N$ has a nonnegative value $v_{i,j}$ for each good $j \in M$, and their utility for a bundle $S \subseteq M$ is additive: $u_i(S) = \sum_{j \in S} v_{i,j}$. An \emph{allocation} is a partition of $M$ into disjoint bundles $A_1, \ldots, A_n$ assigned to the agents.

A fundamental goal is to ensure that the allocation is fair. One of the best known fairness notions is \emph{envy-freeness (EF)}, which requires that no agent prefers someone else’s bundle: $u_i(A_i) \ge u_i(A_j)$ for all $i, j \in N$. However, envy-free allocations are not guaranteed to exist for indivisible goods. For instance, if there is a single indivisible item to be allocated (e.g. a necklace), an agent who does not get the item will envy the one who does.

Since envy-freeness is often unattainable, alternative fairness notions have been proposed that adapt more naturally to discrete settings.
A widely studied notion is \emph{envy-freeness up to one good (EF1)}, which requires that envy can be eliminated by removing some item from the envied bundle. A stronger notion is \emph{envy-freeness up to any good (EFX)}~\cite{Caragiannis19}, which demands that for every pair of agents $i$ and $j$, and for every good $g \in A_j$, agent $i$ does not envy agent $j$ once $g$ is removed: $u_i(A_i) \ge u_i(A_j \setminus \{g\})$.

More generally, an allocation is said to be \emph{$(1-\epsilon)$-EFX}, for some  $\epsilon \in [0,1)$ if for all $i, j \in N$ and all $g \in A_j$, it holds that $u_i(A_i) \ge (1-\epsilon) \cdot u_i(A_j \setminus \{g\})$. For general sub-additive valuations, it is known that $(1/2)$-EFX allocations exist ~\cite{PlautRoughgarden20}, and better approximations are possible for restricted valuation classes.  

The existence of \emph{exact} EFX allocations\footnote{That is, $(1-\epsilon)$-EFX with $\epsilon = 0$.} has been conjectured. Proving or disproving this conjecture represents a major open problem in fair division.
Existence of exact EFX allocations for additive valuations has  been proved for two agents~\cite{PlautRoughgarden20}, three agents~\cite{chaudhury2020efx}, and certain special classes of  valuations (see, e.g., \cite{Amanatidis20,HVGN025}).

\subsection{Our Contribution.} 

Our main contribution is to design  a simple local search algorithm for computing EFX allocations on instances with additive valuations. 
The algorithm employs simulated annealing, using  the number of EFX violations as an objective function, and explores the space of allocations through a single-item transfer neighborhood structure.

Our experiments show that the algorithm finds EFX allocations consistently and efficiently across a wide range of instance sizes, with uniform and correlated valuations. In fact, the algorithm converged to an EFX allocation on every instance we tested. We also observe that the runtime is  affected by the item-to-agent ratio $m/n$ and by  how correlated the agents' valuations are. Implementation details and reproducibility resources are described in the experimental results section.

On the theoretical side, we provide a potential function for identical additive valuations, which is defined for  each allocation $A$ as the variance across agents of their total bundle values. We  show that if an allocation $A$ is not EFX, then some single-item transfer strictly decreases $\Phi$. Consequently, any strict-descent process on $\Phi$ with the single-transfer neighborhood structure ends at an EFX allocation, giving an alternative existence proof for identical additive valuations.

\subsection{Related Work.}

The study of fair allocation of indivisible goods has seen significant growth over the last decade, particularly concerning relaxations of envy-freeness such as EF1 and EFX. For a comprehensive survey, see \cite{AmanatidisABFLMVW23}.

\paragraph{EFX existence.}
Envy-freeness up to any good (EFX) has become a central open problem in fair division due to its intuitive appeal. Despite significant attention, the general existence of EFX allocations remains unresolved. Prior work has tackled this challenge by focusing on specific scenarios or relaxed versions of EFX. For example, \cite{PlautRoughgarden20} initially established EFX existence for two agents and for any number of agents with identical (not necessarily additive) valuations. Extending beyond two agents, \cite{chaudhury2020efx} confirmed the existence of EFX allocations for three agents with additive valuations and showed they can be computed in pseudo-polynomial time. \cite{Amanatidis20} showed that EFX allocations exist for bi-valued instances.   \cite{HVGN025} showed that EFX allocations exist for $n$ agents whose valuations fall in one of three types.

Work by \cite{akrami2024efx}  advanced this line of research by proving the existence of EFX allocations in scenarios where two agents possess general monotone valuations, and at least one agent has an MMS-feasible valuation. Their approach simplifies prior methodologies by eliminating the need for complex constructions such as envy or champion graphs.

\paragraph{Approximate EFX.} On another front, approximate EFX allocations with limited unallocated goods (charity) have been extensively explored. Initial work by \cite{DGKPS14} introduced EFX with charity and showed that meaningful approximations are possible. Later, \cite{chaudhury2021little} leveraged combinatorial arguments based on the ``rainbow cycle number'' (RCN) to achieve $(1-\varepsilon)$-EFX allocations with bounded charity. \cite{berger2022almost} show that every 4-agent instance with additive valuations has an EFX allocation  that leaves at most a single item unallocated, and every setting with $n$ additive valuations has an EFX allocation with at most $n-2$ unallocated items. \cite{berger2022almost} extend these results to nice cancelable valuations (which include  additive, unit-demand, budget-additive and multiplicative valuations). \cite{GHLVV23} analyze submodular valuations and design a polynomial time algorithm based on local search to find  an allocation that is simultaneously 1/2-EFX and approximates the optimal Nash social welfare within a  factor of $(8+\epsilon)$ for every $\epsilon > 0$.

\paragraph{EF1 and proportionality.} We briefly mention several other prominent fairness notions for indivisible goods. Envy-freeness up to one good (EF1) is guaranteed to exist for arbitrary monotone valuations~\cite{lipton2004approximately} and has been extensively studied due to its practical applicability and computational tractability. Reachability of EF1 allocations via sequential exchanges was studied in \cite{IKSY24}, where the question is whether one EF1 allocation can be reached from another EF1 allocation via a sequence of exchanges such that every intermediate allocation is  EF1. EF1 allocations for mixed (divisible and indivisible) goods has  been studied in \cite{BEI2021103436}.

Proportionality requires that each agent has utility at least $1/n$ of the utility they would have if they received the entire set of items. While proportionality cannot always be ensured for indivisible goods, relaxations of it can. There exist  various proportionality relaxations such as Maximin share (MMS), Proportionality up to one good (PROP1), and Proportionality up to any good (PROPx). Each of these has different guarantees and computational complexities. \cite{CSegalWarut_weighted_fairness} study weighted versions of these fairness notions.

\paragraph{Competitive equilibrium, share-based notions, and stochastic settings.} The competitive equilibrium from equal incomes has also been considered as a method for allocating indivisible goods (see, e.g., \cite{budish2011combinatorial,BL14,OPR14,BHM15}), with extensions to competitive equilibria from unequal budgets \cite{babaioff2021competitive} or random incomes \cite{CERI25}. Share based notions of fairness for indivisible goods have also been proposed  (see, e.g., \cite{babaioff2021competitive,BF22,BFRV24}).

The existence of fair allocations in stochastic settings (e.g. with random valuations) has also been studied (see, e.g., \cite{DGKPS14,manurangsi2020envy,Psomas_EFX_stochastic}). 

\paragraph{Local search in other domains.} Stochastic local search methods have also been extensively studied in domains outside fair division, such as Boolean satisfiability. For example, the WalkSAT algorithm for satisfiability~\cite{papadimitriou1991selecting,schoning1999probabilistic} selects an  unsatisfied clause uniformly at random and flips the value of a random variable  from that clause.  In \cite{selman1993local} the algorithm additionally performs  a random walk over variables that appear in unsatisfied clauses to escape local minima.

\section{Definitions.}

In order to introduce the algorithm, we first define the notion of EFX violation.

\begin{definition}[\bf EFX violation]
Let  $A = (A_1, \ldots, A_n)$ be an allocation. A tuple  $(i,j,g)$, where $i,j \in N$ and $g \in M$, represents an EFX violation if $u_i(A_j \setminus \{g\}) > u_i(A_i)$.
\end{definition}

 Our measure of progress  is the total number of EFX violations.

\begin{definition}[\bf EFX violation count] \label{def:main_objective}
Given an allocation $A = (A_1, \ldots, A_n)$, the EFX violation count at $A$ is: 
\[
f(A) = \sum_{i \in N} \sum_{j \in N \setminus \{i\}} \sum_{g \in A_j} c(i,j,g), 
\]
where $
c(i,j,g) = \mathbb{1}\{u_i(A_j \setminus \{g\}) > u_i(A_i)\} \,.$
\end{definition}
Clearly, an allocation $A$ is EFX if and only if $f(A)=0$. 

We also define a neighborhood structure in the space of allocations.

\begin{definition}[\bf Single-transfer neighborhood] \label{def:neighbor}
Let $A = (A_1, \ldots, A_n)$ be an allocation.
An allocation $A' = (A_1', \ldots, A_n')$ is a \emph{neighbor} of $A$ if it differs from $A$ by the transfer of a single item from its current owner in $A$ to another agent. 
\end{definition}

\section{Algorithm.}

Our algorithm explores the space of allocations via local moves, starting from an initial allocation $A$ drawn uniformly at random.

At each step, the algorithm generates a random neighbor $A'$ of the current allocation $A$ by selecting an item $j$ uniformly at random and reassigning it from its owner $i$ in $A$ to another agent chosen uniformly at random from $N \setminus \{i\}$. It then compares the number of EFX violations at $A$ and $A'$:
\begin{itemize}
    \item If $f(A') < f(A)$, the algorithm moves to $A'$.
    \item Otherwise, it moves to $A'$ with a small probability $p = \exp(-\Delta / T)$, where $\Delta = f(A') - f(A)$ and $T$ is a temperature parameter that decreases over time.
\end{itemize}
This process repeats until an allocation with zero violations is found.

The non-zero transition probability to allocations with higher $f$ values is necessary: there exist instances where the function $f$ has strictly positive local minima.

\paragraph{Pseudocode.} Below, we present the pseudocode for our algorithm.
Parameter values used in the implementation (such as the initial temperature $T$) are shown in brackets when introducing each parameter.

\begin{verbatim}
Input: 
    n - number of agents
    m - number of items
    v - valuation matrix, where v[i][j] is agent i's value for item j
Output:
    An EFX allocation (if one is found)

1. Repeat until an EFX allocation is found:

    a. Initialize a random allocation A, by assigning each item to a uniformly
    random agent.
    
    b. Set the temperature T to an initial value (e.g., 5.0) and define a minimum 
    threshold (e.g., Tmin = 0.0001).

    c. While T > Tmin and the allocation is not yet EFX:
    
        i. For a number of steps at the current temperature (e.g., 100 × n × m):
            
            - Count the EFX violations at the current allocation (i.e., f(A)).
            
            - Propose a new allocation A' by selecting an item uniformly at random, 
            removing it from its current owner in A, and reassigning it to a 
            different agent (chosen uniformly at random).
            
            - Let D be the change in the number of EFX violations if the item is
            reassigned (i.e., D = f(A') - f(A))
            
            - If D < 0 (i.e., fewer violations) accept the move  (i.e., update the 
            allocation to A').
            
            - Else, accept the move with probability exp(-D / T).

        ii. Decrease the temperature  (e.g., T ← 0.99 × T).

2. Return the EFX allocation once one is found (i.e., the number of violations
is zero).
\end{verbatim}

\section{Experimental Results.}

We begin with uniform random instances, where each entry $v_{i,j}$ of the valuation matrix  is drawn independently and uniformly from the interval $[0,1]$.

The key goals are to evaluate whether the algorithm can find EFX allocations and measure the runtime required for convergence as a function of the number of agents and items.
To the best of our knowledge, no efficient, publicly available algorithm exists for computing EFX allocations in these settings. For context, the search space for $n=4$ and $m=1000$ is already computationally prohibitive for  exhaustive search methods.

\paragraph{Experimental environment.} Experiments were conducted within a KVM virtual machine running Debian GNU/Linux (kernel 6.12.27). The virtual machine was provisioned with 24 threads from an AMD EPYC 7B12 host processor and 94 GB of RAM. 
The  code was written in Python 3.13.5 and is available at \cite{code_github}.

\subsection{Scaling with the Number of Items (Fixed Agents)}

In this section we check the performance of the algorithm for various combinations of agents ($n$) and items ($m$).

Table~\ref{tab:combined_performance} summarizes these initial results, presenting the average runtime (left) and the corresponding number of algorithmic steps (right) for several configurations. The ``number of steps'' refers to the total count of neighboring allocations generated before an EFX allocation is found.

An  observation from Table~\ref{tab:combined_performance} is that the runtime is not monotonic with respect to the number of items. For example, when $n=15$, the algorithm is substantially slower for $m=150$ items than for  $m=750$ items. This finding suggests that the performance of the algorithm depends more critically on the item-to-agent ratio ($m/n$) than on the absolute number of items. To investigate this phenomenon more closely, the next section provides a detailed analysis focused on the $m/n$ ratio for a fixed number of agents.

\begin{table}[h!]
    \centering
    \begin{small}

    \begin{subtable}[t]{.48\textwidth}
        \centering
        \begin{tabular}{|c||c|c|c|c|}
        \hline
        \multirow{2}{*}{\makecell{Agents\\ ($n$)}} & \multicolumn{4}{c|}{Items} \\
        \cline{2-5}
        & 1000 & $10n$ & $50n$ & $5 n^2$ \\
        \hline \hline 
        4  & \makecell{0.41 \\ $\pm$ 0.24} & \makecell{0.004 \\ $\pm$ 0.003} & \makecell{0.02 \\ $\pm$ 0.01} & \makecell{0.007 \\ $\pm$ 0.004} \\
        \hline
        10 & \makecell{1.46 \\ $\pm$ 0.53} & \makecell{0.34 \\ $\pm$ 0.26} & \makecell{0.42 \\ $\pm$ 0.13} & \makecell{0.42 \\ $\pm$ 0.13} \\
        \hline
        15 & \makecell{2.22 \\ $\pm$ 0.63} & \makecell{15.24 \\ $\pm$ 18.72} & \makecell{1.34 \\ $\pm$ 0.43} & \makecell{2.75 \\ $\pm$ 0.94} \\
        \hline
        20 & \makecell{2.73 \\ $\pm$ 0.67} & \makecell{599.70 \\ $\pm$ 405.04} & \makecell{2.73 \\ $\pm$ 0.67} & \makecell{8.63 \\ $\pm$ 2.37} \\
        \hline
        \end{tabular}
        \caption{Average runtime (in seconds).} 
        \label{tab:runtime_small_m}
    \end{subtable}
    \hfill
    \begin{subtable}[t]{.48\textwidth}
        \centering
        \begin{tabular}{|c||c|c|c|c|}
        \hline
        \multirow{2}{*}{\makecell{Agents\\ ($n$)}} & \multicolumn{4}{c|}{Items} \\
        \cline{2-5}
        & 1000 & $10n$ & $50n$ & $5 n^2$ \\
        \hline \hline 
        4  & \makecell{797 \\ $\pm$ 472} & \makecell{135  \\ $\pm$ 79} & \makecell{232 \\ $\pm$ 136} & \makecell{135 \\ $\pm$ 71} \\
        \hline
        10 & \makecell{1902 \\ $\pm$ 697} & \makecell{3485 \\ $\pm$ 2737} & \makecell{1169 \\ $\pm$ 375} & \makecell{1169  \\ $\pm$ 375} \\
        \hline
        15 & \makecell{2544 \\ $\pm$ 714} & \makecell{99186 \\ $\pm$ 121372} & \makecell{2177 \\ $\pm$ 694} &  \makecell{2802 \\ $\pm$ 948}\\
        \hline 
        20 & \makecell{3181 \\ $\pm$ 780} & \makecell{2941904 \\ $\pm$ 1990182} & \makecell{3181 \\ $\pm$ 780} & \makecell{4783 \\ $\pm$ 1324} \\
        \hline 
        \end{tabular}
        \caption{Average number of steps.} 
        \label{tab:runtime_small_m_steps}
    \end{subtable}
    
    \end{small}
    
    \caption{Performance metrics for finding EFX allocations across various problem sizes. Each entry represents the average over 100 instances with agent valuations drawn uniformly at random. The value below each mean indicates one standard deviation.}
    \label{tab:combined_performance}
\end{table}

\subsection{Effect of the Item-to-Agent Ratio}

To investigate the effect of item availability, we conducted an experiment with a fixed number of agents ($n=15$) and varied the number of items from $m=3$ to $m=1050$. This design allowed us to observe the algorithm's performance as it transitioned from a state of item scarcity ($m < n$) to one of item abundance ($m \gg n$).
The evolution of the average runtime as the number of items grows is plotted  in Figure~\ref{fig:visualization_15_agents_critical_region_runtime}. The corresponding plot for the step count is in the appendix (Figure~\ref{fig:visualization_15_agents_critical_region_steps}). 

The primary finding is that the problem is hardest for the algorithm not when the number of items is largest, but within a volatile region at low item-to-agent ratios. As shown in Figure~\ref{fig:visualization_15_agents_critical_region_runtime}, the runtime does not exhibit a single sharp peak; instead, it oscillates with the maximum trending upward, reaching its global maximum near $m/n \approx 3.5$.


Beyond this hard region, as the number of items increases further, the problem becomes easier, and the runtime dramatically decreases to a minimum around $m/n \approx 25$. Subsequently, the runtime begins to increase again, but very slowly.
This non-monotonic behavior strongly suggests that the  hardness is driven by the problem's combinatorial structure rather than merely the size of the input.

\begin{figure}
\centering
    \begin{subfigure}[b]{\textwidth}
        \centering
        \includegraphics[width=\textwidth]{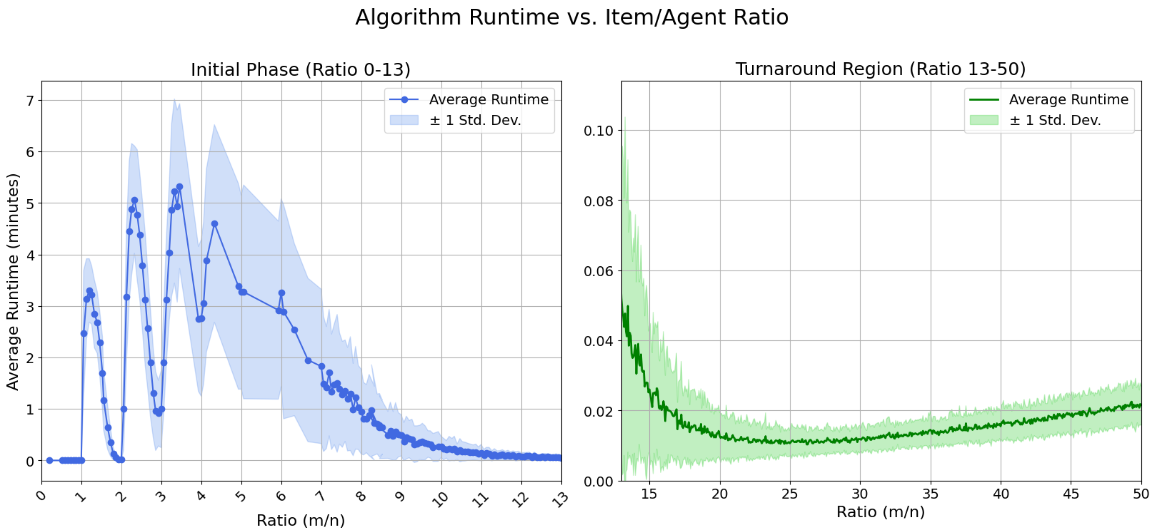} 
        \label{fig:15_agents_closeup}
    \end{subfigure}\\ 
    \begin{subfigure}[b]{0.55\textwidth}
        \centering
        \includegraphics[width=\textwidth]{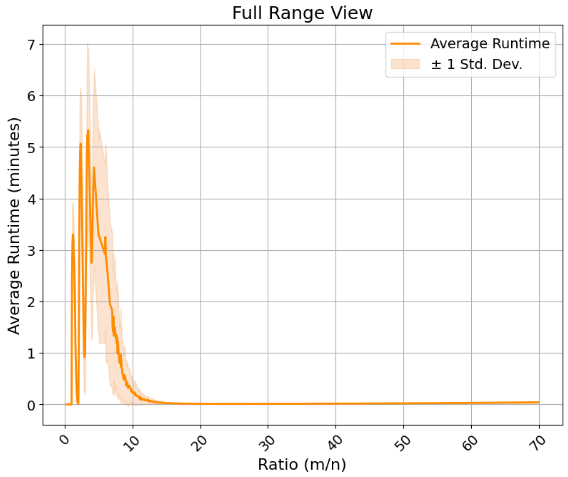} 
        \label{fig:15_agents_fullview}
    \end{subfigure}
\caption{Average runtime (in minutes) vs. the item-to-agent ratio 
($m/n$) for finding an EFX allocation with 
$n=15$ agents. Each data point shows the average runtime over 100 independent instances where agent valuations were drawn uniformly at random. The shaded region indicates one standard deviation. The figure presents three panels: a close-up of the region $m/n \in [0, 13]$, a second close-up of the region $m/n \in [13, 50]$, and a full-range view.}
\label{fig:visualization_15_agents_critical_region_runtime}
\end{figure}

Prior research from \cite{DGKPS14} has shown  that for random valuations, an exact envy-free allocation exists with high probability when $m \in \Omega(n \log{n})$ and is  unlikely to exist when $m \in n + o(n)$. Follow-up work by \cite{manurangsi2020envy} improved the existence bound from \cite{DGKPS14}, showing that envy-free allocations exist with high probability when $m \geq 2n$ and $n$ divides $m$.

Understanding whether this computational barrier stems from the structure of the problem or the nature of the algorithm could shed light on the broader complexity of computing EFX allocations.
For example, when $m=n+1$, the problem can in fact be solved efficiently using a round-robin approach: let the agents, in the order $1, 2, \ldots, n-1, n, n$, choose their favorite items among the remaining ones. The resulting allocation is EFX.

\subsection{Scaling with the Number of Agents (Fixed Items)}

We next evaluate the algorithm's scalability with respect to the number of agents, $n$. For this experiment, we fixed the number of items to a large value ($m=10^4$) and increased the number of agents from $n=4$ to $n=100$. The results are shown  in Figure~\ref{fig:performance_10K_items}. 

\medskip 
\begin{figure}[h!]
    \centering
    \includegraphics[scale=0.86]{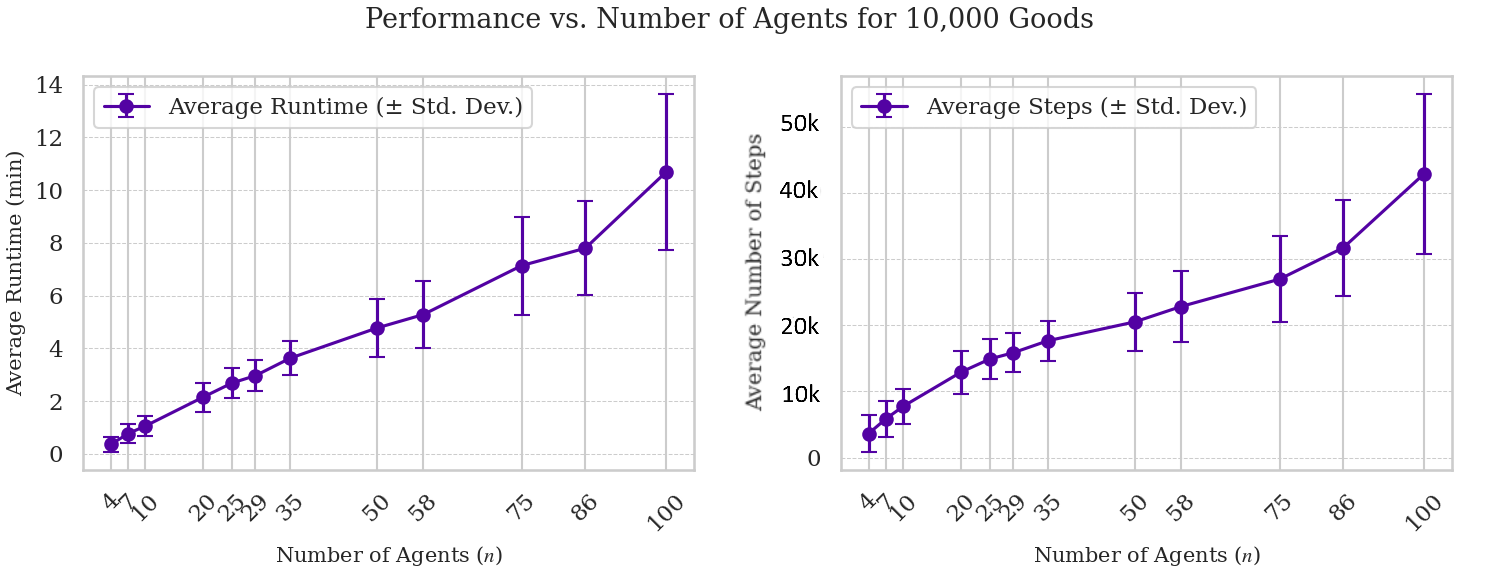}
    \caption{Performance of the algorithm with $m = 10,000$ items as the number of agents grows from $n=4$ to $n=100$. For each value of $n$ displayed, the data point is obtained by taking  the average of 100 runs with uniformly random valuations. The left panel shows the average runtime and standard deviation (in minutes), while the right panel shows the average and standard deviation in steps.}
    \label{fig:performance_10K_items}
\end{figure}

\subsection{Correlated Valuations}

We also test correlated valuations using  a simple ``shared value + individual term'' model
parametrized by a correlation strength $\rho \in [0,1]$.

For each item $j \in [m]$ we draw a common value  $w_j \sim \mathrm{Unif}[0,1]$. 
Independently, for each agent $i \in [n]$  we also draw an individual term  $u_{i,j} \sim \mathrm{Unif}[0,1]$. 
The value of agent $i$ for item $j$ is then defined as 
\[
v_{i,j} \;=\;   \rho\,w_j + (1-\rho)\,u_{i,j}\,.
\]
This convex combination keeps $v_{i,j}\in[0,1]$, gives i.i.d.\ uniform valuations when $\rho=0$, 
and yields identical valuations across agents when $\rho=1$. 
Intermediate $\rho$ smoothly interpolates between heterogeneous and identical preferences.

\paragraph{Empirical behavior.} In our experiments, the effect of the correlation strength $\rho$ depends on the item-to-agent ratio $m/n$.

We first test an item-abundant regime with $n=8$ and $m=160$, corresponding to  a ratio $m/n =20$. This is a setting where instances with uniform valuations are computationally straightforward. For correlated valuations, the mean runtime increases with the correlation strength $\rho$, peaking sharply near $\rho =0.95$ due to a heavy-tailed runtime distribution, before dropping again for identical valuations ($\rho=1$).
See Figure~\ref{fig:visualization_8_agents_160_items_correlated} for runtimes, with full numeric values in the appendix (Table~\ref{tab:correlated_runtimes} and \ref{tab:correlated_steps}).

For our second case, we select parameters $n=15$ and $m=52$. This creates an item-scarce regime ($m/n \approx 3.46$) specifically chosen because this exact setting represents one of the most computationally challenging peaks for uniform valuations (as shown in Figure \ref{fig:visualization_15_agents_critical_region_runtime}). 
Here the trend is reversed for small $\rho$, with the hardest regime occurring at $\rho=0$. See 
Figure~\ref{fig:visualization_15_agents_52_items_correlated} for runtimes, with full numeric values in the appendix (Table~\ref{tab:correlated_runtimes_15_agents_52_items} and \ref{tab:correlated_steps_15_agents_52_items}).

\begin{figure}[h!]
\centering
    \begin{subfigure}[b]{0.48\textwidth}
        \centering
        \includegraphics[scale=0.46]{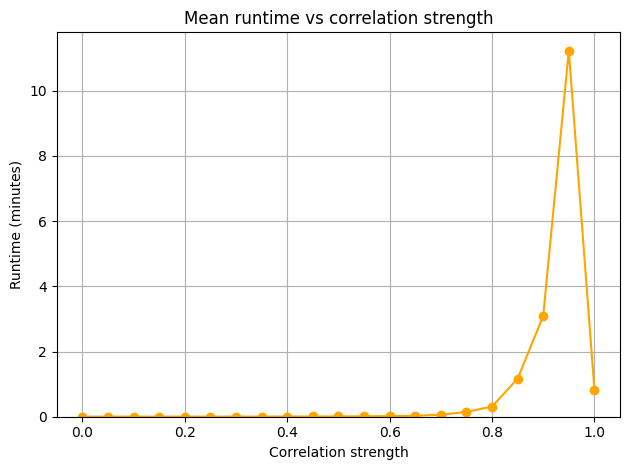}
        \label{fig:8_agents_160_items_mean}
    \end{subfigure}
    \begin{subfigure}[b]{0.48\textwidth}
        \centering
        \includegraphics[scale=0.46]{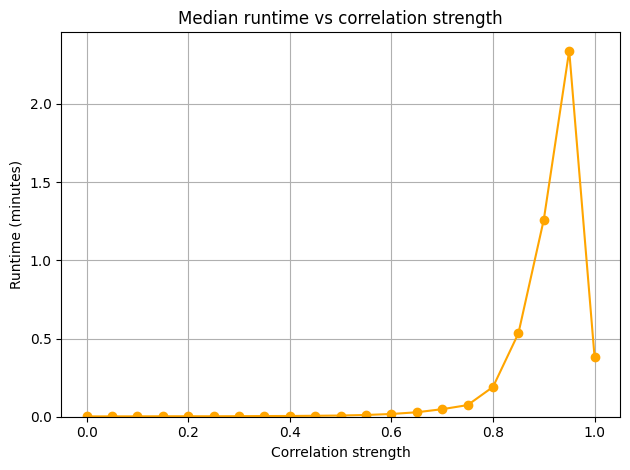} 
        \label{fig:8_agents_160_items_median}
    \end{subfigure}
    \caption{Runtime (in minutes) for finding an EFX allocation with $n=8$ agents and $m=160$ items as the correlation strength grows from 0 to 1.  For each correlation strength, the results are averaged over 100 trials. The left figure shows the mean and the right figure shows the median.}
\label{fig:visualization_8_agents_160_items_correlated}
\end{figure}

\begin{figure}[h!]
\centering
    \begin{subfigure}[b]{0.48\textwidth}
        \centering
        \includegraphics[scale=0.46]{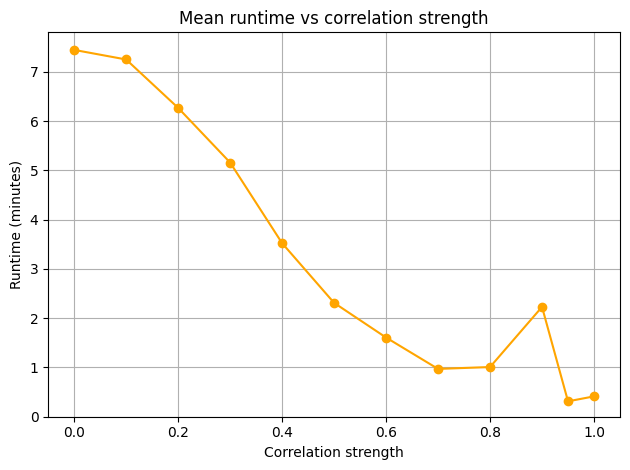}
        \label{fig:15_agents_52_items_mean}
    \end{subfigure}
    \begin{subfigure}[b]{0.48\textwidth}
        \centering
        \includegraphics[scale=0.46]{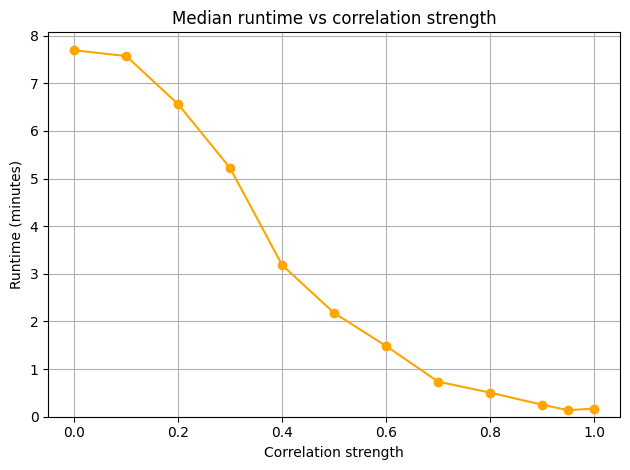} 
        \label{fig:15_agents_52_items_median}
    \end{subfigure}
    \caption{Runtime (in minutes) for finding an EFX allocation with $n=15$ agents and $m=52$ items as the correlation strength grows from 0 to 1.  For each correlation strength, the results are averaged over 100 trials. The left figure shows the mean and the right figure shows the median.}
\label{fig:visualization_15_agents_52_items_correlated}
\end{figure}

\section{Potential Function for  Identical Valuations} \label{sec:identical}

When valuations are identical, each good $g\in M$ has a common value $w_g$; that is $v_{i,g}=w_g$  $\forall i\in N$. 

We find a potential function $\Phi$ that assigns a value to each allocation $A$ and has the property that, if $A$ is not EFX, there exists a single-item transfer that can be made to strictly decrease $\Phi$. Consequently, any strict-descent procedure on $\Phi$ with the single-transfer neighborhood terminates at an EFX allocation. While the existence of EFX allocations is known for identical additive valuations from \cite{PlautRoughgarden20}, to the best of our knowledge this proof is new.

\begin{proposition} \label{prop:identical_valuations}
Consider an instance with identical valuations, where each good $g \in M$ has a common value $w_g > 0$.
For an allocation $A = (A_1, \ldots, A_n)$, define 
\begin{align}  \label{eq:def_Phi_and_aux}
\Phi(A) \;:=\; \sum_{i=1}^n \left( Y_i - \mu \right)^2, \quad \text{where } Y_i := \sum_{g\in A_i} w_g \; \mbox{  and  } \; \mu := \frac{1}{n}\sum_{g\in M} w_g \,. 
\end{align}
Then every strict descent procedure on $\Phi$ 
using the single-good transfer neighborhood from \cref{def:neighbor} is  guaranteed to terminate at an EFX allocation.
\end{proposition}
\begin{proof}
An allocation $A$ with identical valuations is EFX if and only if 
for each pair of agents $i \neq j$ and all $g \in A_j$, we have $u_i(A_i) \geq u_i(A_j \setminus \{g\})$. Since the goods have common value, the EFX condition is equivalent to 
\begin{align}  \label{eq:EFX_condition_identical_valuations}
\sum_{h \in A_i} w_h \geq \bigl( \sum_{h \in A_j} w_h\bigr) - w_g \,.
\end{align}
Using the definition of $Y_k$ from \eqref{eq:def_Phi_and_aux}, inequality \eqref{eq:EFX_condition_identical_valuations} can be rewritten as  $Y_j - Y_i \leq w_g$.

Consider an allocation $A$ that is not EFX. Then there exists a triple $(i, j, g)$ with $i,j \in N$ and  $g \in A_j$ such that $Y_j - Y_i > w_g$. Let $A'$ be the new allocation obtained from $A$ by transferring good $g$ from agent $j$ to agent $i$. The change in the objective $\Phi$ is:
\[
\Delta\Phi = \Phi(A') - \Phi(A) = 2w_g(w_g - (Y_j - Y_i)).
\]
Since $w_g > 0$ and we assumed $Y_j - Y_i > w_g$, it follows that $w_g - (Y_j - Y_i) < 0$. Thus $\Delta\Phi < 0$, and so $\Phi(A') < \Phi(A)$.

This shows that the existence of any EFX violation at allocation $A$ guarantees that a strictly improving move (i.e. single-item transfer) is available from $A$.  A strict descent algorithm on a finite state space is guaranteed to terminate in a local minimum. Since non-EFX allocations cannot be local minima, the algorithm must terminate at an EFX allocation.
\end{proof}

\section{Discussion.}

\paragraph{EFX existence.} Our algorithm found an EFX allocation in all instances tested. This suggests not only that EFX allocations are likely to  exist for additive valuations on the distributions tested, but  that they are also computationally accessible via simple local search dynamics.
It would be interesting to better understand the computational hardness of the EFX problem as a function of the ratio $m/n$ and the degree to which the valuations are correlated.

\paragraph{Accelerating convergence.} One could investigate the effect of warm starts on the algorithm's performance. A natural choice for initialization is the welfare-maximizing allocation, where the {social welfare} of an allocation \( A \) is defined as \( SW(A) = \sum_{i \in N} u_i(A_i) \). Prior work~\cite{DGKPS14} shows that when the number of goods satisfies \( m \in \Omega(n \log n) \), the welfare-maximizing allocation is envy-free (and thus EFX) with high probability.

We observed that when items are abundant, initializing the algorithm from a welfare-maximizing allocation leads to faster convergence than starting from a random one, even when the initial allocation is not  EFX; the algorithm corrects the remaining violations quickly. While we do not report detailed experiments on this aspect, it suggests that high-welfare initializations may serve as effective warm starts when the number of items is large.

Other potential directions include exploring different neighborhood structures (e.g., 2-local moves), experimenting with alternative objective functions and annealing parameters (such as temperature decay schedules), and adding a regularization term to the objective.

\paragraph{Comparison with round robin.} We also compared to the round robin algorithm, which is a sequential method where agents take turns making a greedy choice---selecting their most valued item from the available pool---while following a fixed, cyclical order (e.g., $1, \ldots, n$) until all items are allocated. Round robin  fails to produce an EFX allocation in many instances, especially for correlated valuations.

\section{Acknowledgements} 
The author would like to thank Ruta Mehta for helpful feedback. The author also acknowledges using large language models (ChatGPT and Gemini) while brainstorming to explore candidate objective/potential functions.

\bibliographystyle{alpha}
\bibliography{EFX_bib}

\newpage 

\appendix  

\section{Appendix}

\begin{figure}
\centering
    \begin{subfigure}[b]{\textwidth}
        \centering
        \includegraphics[width=\textwidth]{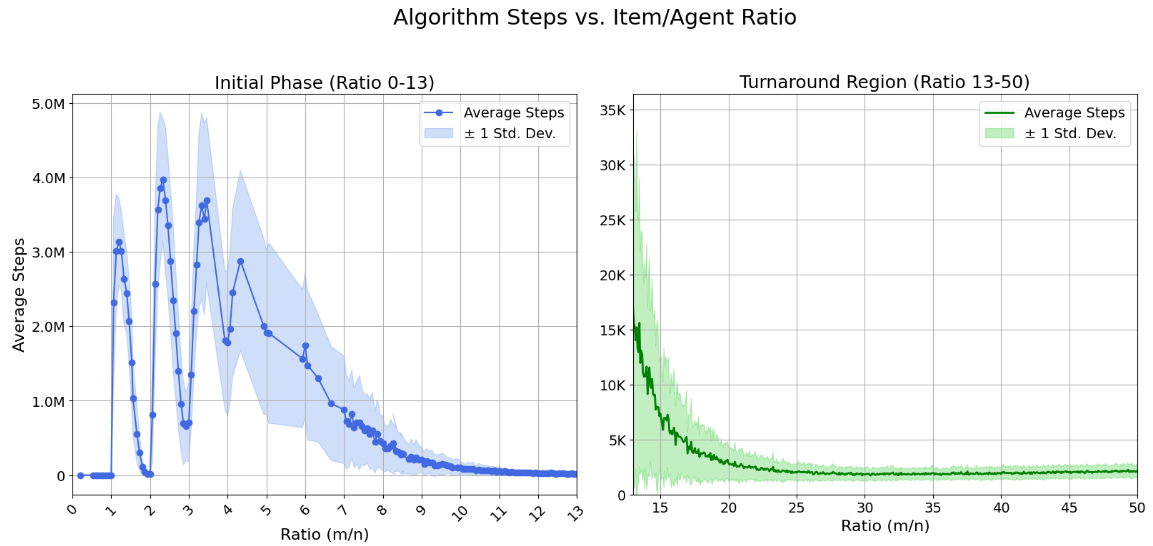} 
        \label{fig:15_agents_closeup_steps}
    \end{subfigure}\\ 
    \begin{subfigure}[b]{0.55\textwidth}
        \centering
        \includegraphics[width=\textwidth]{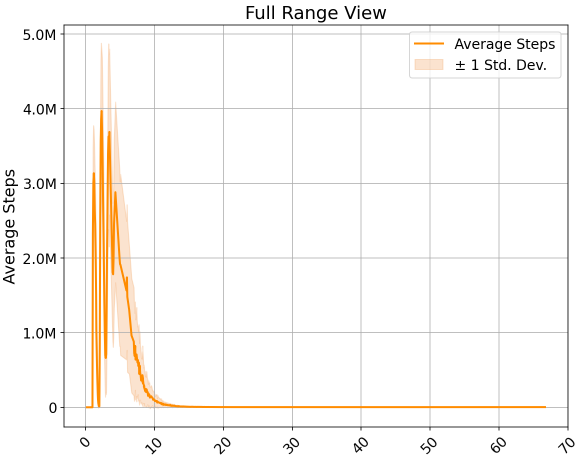} 
        \label{fig:15_agents_fullview_steps}
    \end{subfigure}
\caption{Average number of steps  vs. the item-to-agent ratio 
($m/n$) for finding an EFX allocation with 
$n=15$ agents. A step counts a neighboring allocation of the current one that was generated by the algorithm. Each data point shows the average number of steps taken by the algorithm over 100 independent instances where agent valuations were drawn uniformly at random. The shaded region indicates one standard deviation. The figure presents three panels: a close-up of the region $m/n \in [0, 13]$, a second close-up of the region $m/n \in [13, 50]$, and a full-range view.}
\label{fig:visualization_15_agents_critical_region_steps}
\end{figure}

In the next two tables we include the runtimes and number of steps required to find an EFX allocation for $n=8$ agents and $m=160$ items with correlated valuations.
\begin{table}[h!]
    \centering
    \footnotesize 
    \setlength{\tabcolsep}{3pt} 
    
    \newcommand{\headerstrut}{\rule[-.5em]{0pt}{1.5em}}

    \begin{tabular}{|l||c|c|c|c|c|c|c|c|c|c|c|c|}
        \hline
        & \multicolumn{12}{c|}{\headerstrut Correlation Strength} \\
        \cline{2-13}
         Runtimes (s) & \headerstrut 0.0 & 0.1 & 0.2 & 0.3 & 0.4 & 0.5 & 0.6 & 0.7 & 0.8 & 0.9 & 0.95 & 1.0 \\
        \hline
        \headerstrut \makecell[l]{Mean  $\pm$ SD} & \makecell{0.13\\$\pm$0.05} & \makecell{0.16\\$\pm$0.07} & \makecell{0.19\\$\pm$0.07} & \makecell{0.24\\$\pm$0.11} & \makecell{0.34\\$\pm$0.19} & \makecell{0.54\\$\pm$0.36} & \makecell{1.22\\$\pm$0.81} & \makecell{3.60\\$\pm$2.99} & \makecell{18.69\\$\pm$20.83} & \makecell{185.55\\$\pm$274.18} & \makecell{673.91\\ $\pm$ 3457.02}  & \makecell{49.88 \\ $\pm$ 63.43} \\
        \hline 
        \headerstrut \makecell[l]{Median} & 0.12 & 0.14 & 0.18 & 0.21 & 0.29 & 0.46 & 1.05 & 2.89 & 11.44 & 75.62 & 140.41 & 22.75 \\
        \hline
    \end{tabular}
    \caption{
        Runtime (in seconds) for finding an EFX allocation with $n=8$ agents and $m=160$ items. The first data row shows the mean $\pm$ std. dev.; the second shows the median. For each correlation strength, the results are averaged over 100 trials.    }
    \label{tab:correlated_runtimes} 
\end{table}

\begin{table}[h!]
    \centering
    \footnotesize
    \setlength{\tabcolsep}{3pt}
    
    \newcommand{\headerstrut}{\rule[-.5em]{0pt}{1.5em}}

    \begin{tabular}{|l||c|c|c|c|c|c|c|c|c|c|c|c|}
        \hline
        & \multicolumn{12}{c|}{\headerstrut Correlation Strength} \\
        \cline{2-13}
         Steps & \headerstrut 0.0 & 0.1 & 0.2 & 0.3 & 0.4 & 0.5 & 0.6 & 0.7 & 0.8 & 0.9 & 0.95 & 1 \\
        \hline
        \headerstrut \makecell[l]{Mean \\ $\pm$ SD} & \makecell{540\\$\pm$205} & \makecell{647\\$\pm$282} & \makecell{783\\$\pm$306} & \makecell{978\\$\pm$463} & \makecell{1422\\$\pm$784} & \makecell{2270\\$\pm$1493} & \makecell{5093\\$\pm$3399} & \makecell{15002\\$\pm$12464} & \makecell{77906\\$\pm$86746} & \makecell{774k\\$\pm$1.15M} &  \makecell{2.81M\\ $\pm$ 14.44M} & \makecell{205k \\ $\pm$ 263k } \\
        \hline 
        \headerstrut \makecell[l]{Median} & 483 & 563 & 742 & 889 & 1198 & 1871 & 4352 & 12221 & 48082 & 314186 & 585780 & 90008 \\
        \hline
    \end{tabular}
    \caption{
        Number of algorithmic steps (neighboring allocations generated) required to find an EFX allocation with $n=8$ agents and $m=160$ items. The first data row shows the mean $\pm$ std. dev.; the second shows the median. For each correlation strength, the results are averaged over 100 trials. We use 'k' for thousands and 'M' for millions.
    }
    \label{tab:correlated_steps}
\end{table}

In the next two tables we include the runtimes and number of steps required to find an EFX allocation for $n=15$ agents and $m=52$ items with correlated valuations.

\begin{table}[h!]
    \centering
    \footnotesize 
    \setlength{\tabcolsep}{3pt} 
    
    \newcommand{\headerstrut}{\rule[-.5em]{0pt}{1.5em}}

    \begin{tabular}{|l||c|c|c|c|c|c|c|c|c|c|c|c|}
        \hline
        & \multicolumn{12}{c|}{\headerstrut Correlation Strength} \\
        \cline{2-13}
         Runtimes & \headerstrut 0.0 & 0.1 & 0.2 & 0.3 & 0.4 & 0.5 & 0.6 & 0.7 & 0.8 & 0.9 & 0.95 & 1.0 \\
        \hline
        \headerstrut \makecell[l]{Mean  \\ $\pm$ SD} & \makecell{446.83\\$\pm$151.79} & \makecell{435.16\\$\pm$152.13} & \makecell{376.23\\$\pm$160.83} & \makecell{309.64\\$\pm$159.77} & \makecell{211.43\\$\pm$134.80} & \makecell{138.72\\$\pm$95.93} & \makecell{96.46\\$\pm$73.91} & \makecell{58.25\\$\pm$53.93} & \makecell{60.52\\$\pm$78.80} & \makecell{133.79\\$\pm$1093.61} & \makecell{18.75\\ $\pm$ 26.64}  & \makecell{24.79 \\ $\pm$ 38.23} \\
        \hline 
        \headerstrut \makecell[l]{Median} & 461.51 & 454.35 & 393.66 & 313.25 & 191.05 & 130.71 & 89.08 & 44.21 & 30.35 & 15.16 & 8.20 & 10.28 \\
        \hline
    \end{tabular}
    \caption{
       Runtime (in seconds) for finding an EFX allocation with $n=15$ agents and $m=52$ items. The first data row shows the mean $\pm$ std. dev.; the second shows the median. For each correlation strength, the results are averaged over 100 trials.}
    \label{tab:correlated_runtimes_15_agents_52_items} 
\end{table}

\begin{table}[h!]
    \centering
    \footnotesize 
    \setlength{\tabcolsep}{3pt} 
    
    \newcommand{\headerstrut}{\rule[-.5em]{0pt}{1.5em}}

    \begin{tabular}{|l||c|c|c|c|c|c|c|c|c|c|c|c|}
        \hline
        & \multicolumn{12}{c|}{\headerstrut Correlation Strength} \\
        \cline{2-13}
         Steps & \headerstrut 0.0 & 0.1 & 0.2 & 0.3 & 0.4 & 0.5 & 0.6 & 0.7 & 0.8 & 0.9 & 0.95 & 1.0 \\
        \hline
        \headerstrut \makecell[l]{Mean \\ $\pm$ SD} & \makecell{3.41M\\$\pm$ 1.17M} & \makecell{3.29M \\$\pm$ 1.15M} & \makecell{2.94M\\$\pm$ 1.26M} & \makecell{2.41M \\$\pm$ 1.24M} & \makecell{1.64M\\$\pm$1.05M} & \makecell{1.07M\\$\pm$743k} & \makecell{745k\\$\pm$570k} & \makecell{450k\\$\pm$417k} & \makecell{464k\\$\pm$603k} & \makecell{1.02M\\$\pm$8.38M} & \makecell{144k\\ $\pm$ 203k}  & \makecell{190k \\ $\pm$ 294k} \\
        \hline 
        \headerstrut \makecell[l]{Median} & 3.56M & 3.42M & 3.09M & 2.45M & 1.49M & 1.01M & 692k & 340k & 233k & 117k & 63k & 79k \\
        \hline
    \end{tabular}
    \caption{Number of algorithmic steps (neighboring allocations generated) required to find an EFX allocation with $n=15$ agents and $m=52$ items. The first data row shows the mean $\pm$ std. dev.; the second shows the median. For each correlation strength, the results are averaged over 100 trials.  We use 'k' for thousands and 'M' for millions.}
    \label{tab:correlated_steps_15_agents_52_items} 
\end{table}

\end{document}